\documentclass{llncs}

\usepackage[T1]{fontenc}
\usepackage[latin2]{inputenc}
\usepackage[english]{babel}

\usepackage{dsfont}
\usepackage{amsmath, amssymb}
\usepackage{lmodern} 
\usepackage{float, paralist}
\usepackage{tikz, graphicx}
\usetikzlibrary{arrows,decorations.markings,patterns,calc, positioning}

\newcommand{\lex}{{\slshape lex}}
\newcommand{\eg}{{e.\,g.}}
\usepackage[textsize=tiny,textwidth=2.2cm,shadow]{todonotes}

\DeclareMathOperator{\rank}{rank}
\usepackage{algpseudocode}
\usepackage[ruled]{algorithm} 
\begin{document}
\title{Paths to stable allocations}
\author{\'{A}gnes Cseh \and Martin Skutella}

\institute{TU Berlin, Institut f\"ur Mathematik, Stra{\ss}e des 17.~Juni 136, 10623 Berlin, Germany}
\maketitle

\begin{abstract}
The stable allocation problem is one of the broadest extensions of the well-known stable marriage problem. In an allocation problem, edges of a bipartite graph have capacities and vertices have quotas to fill. Here we investigate the case of uncoordinated processes in stable allocation instances. In this setting, a feasible allocation is given and the aim is to reach a stable allocation by raising the value of the allocation along blocking edges and reducing it on worse edges if needed. Do such myopic changes lead to a stable solution?

In our present work, we analyze both better and best response dynamics from an algorithmic point of view. With the help of two deterministic algorithms we show that random procedures reach a stable solution with probability one for all rational input data in both cases. Surprisingly, while there is a polynomial path to stability when better response strategies are played (even for irrational input data), the more intuitive best response steps may require exponential time. We also study the special case of correlated markets. There, random best response strategies lead to a stable allocation in expected polynomial time.

\end{abstract}

\section{Introduction}

Matching markets without prices model various real-life problems such as, \eg, employee placement, task scheduling or kidney donor matching. Research on those markets focuses on maximizing social welfare instead of profit. Stability is probably the most widely used optimality criterion in that case. 

Finding equilibria in markets that lack a central authority of control is another widely studied, challenging task. Besides modeling uncoordinated markets, like third-generation (3G) wireless data networks~\cite{GLM06}, selfish and uncontrolled agents can also represent modifications in coordinated markets, \eg, the arrival of a new participant or slightly changed preferences~\cite{BRR97}. In our present work, those two topics are combined: we study uncoordinated capacitated matching~markets.

\subsection{Stability in matching markets}

The theory of stable matchings has been investigated for decades. Gale and~Shapley~\cite{GS:1962} introduced the notion of stability on their well-known \emph{stable marriage problem}. An instance of this problem consists of a bipartite graph where color classes symbolize men and women, respectively. Each participant has a preference list of their acquaintances of the opposite gender. A set of marriages (a matching) is \emph{stable}, if no pair blocks it. A \emph{blocking pair} is an unmarried pair so that the man is single or he prefers the woman to his current wife and vice versa, the woman is single or she prefers the man to her current husband. The Gale-Shapley algorithm was the first proof for the existence of stable matchings.

A natural extension of matching problems arises when capacities are introduced. The stable allocation problem is defined in a bipartite graph with edge capacities and quotas on vertices. The exact problem formulation and a detailed example are provided~in Section~\ref{sec_pre}.

\subsection{Better and best response steps in uncoordinated markets}
\label{subsec:uncoord}

Central planning is needed in order to produce a stable solution with the Gale-Shapley algorithm. In many real-life situations, however, such a coordination is not available. Agents play their selfish strategy, trying to reach the best possible solution. A \emph{path to stability} is a series of myopic operations. The intuitive picture of a myopic operation is the following. If a man and a woman block a marriage scheme, then they both agree to form a couple together, even if they divorce their current partners to that end. This step may induce new blocking pairs. Such changes are made until a stable matching is reached. 

Note that stability is naturally a desirable property of uncoordinated markets. A stable matching seems to be the best reachable solution for all participants, because they cannot find any partnership that could improve their own position.

The study of uncoordinated matching processes has a long history. In the~case of one-to-one matchings, two different concepts have been studied: better and~best response dynamics. One of the color classes is chosen to be the \emph{active} side. These vertices submit proposals to the \emph{passive} vertices. According to \emph{best response dynamics}, the best blocking edge of an active vertex is chosen to perform~myopic changes along. In \emph{better response dynamics}, any blocking edge can play this~role. 

The core questions regarding uncoordinated processes rise naturally. Can a series of myopic changes result in returning back to the same unstable matching? If yes, is there a way to reach a stable solution? How do random procedures behave? The first question about uncoordinated two-sided matching markets was brought up by Knuth~\cite{Knuth:1976:MSR} in~1976. He also gives an example of a matching problem where better response dynamics cycle. More than a decade later, Roth and Vande Vate~\cite{roth_vandevate} came up with the next result on the topic. They show that random better response dynamics converge to a stable matching with probability one. Analogous results for best response dynamics were published in 2011 by Ackermann et al.~\cite{ackermann2011uncoordinated}. They also show an instance in which best response dynamics cycle give a deterministic algorithm for reaching a stable solution in polynomial time and prove that the convergence time is exponential in both random cases.

\subsubsection{Example for a best response cycle on a matching instance~\cite{ackermann2011uncoordinated}.}

Starting with the unstable matching $(j_2m_2, j_3m_3)$, and saturating the blocking edges $j_1m_3$, $j_2m_1$, $j_3m_1$, $j_1m_2$, $j_2m_2$, $j_3m_3$ in this order leads back to the same unstable matching. In each round, the chosen blocking edge is the best blocking edge of its job.
\begin{figure}[H]
\begin{center}
\begin{tikzpicture}[scale=1, transform shape]

\tikzstyle{vertex} = [circle, draw=black]
\tikzstyle{edgelabel} = [circle, fill=white]

\node[vertex] (j_3) at (1, 4) {$j_3$};
\node[vertex] (j_2) at (-4, 4) {$j_2$};
\node[vertex] (j_1) at (-9, 4) {$j_1$};

\node[vertex] (m_1) at (-9, 0) {$m_1$};
\node[vertex] (m_2) at (-4, 0) {$m_2$};
\node[vertex] (m_3) at (1, 0) {$m_3$};

\draw [] (j_1) -- node[edgelabel, very near start] {3} node[edgelabel, very near end] {1} (m_1);
\draw [] (j_1) -- node[edgelabel, very near start] {1} node[edgelabel, very near end] {2} (m_2);
\draw [] (j_1) -- node[edgelabel, very near start] {2} node[edgelabel, very near end] {1} (m_3);

\draw [] (j_2) -- node[edgelabel, very near start] {1} node[edgelabel, very near end] {3} (m_1);
\draw [ultra thick] (j_2) -- node[edgelabel, very near start] {2} node[edgelabel, very near end] {1} (m_2);
\draw [] (j_2) -- node[edgelabel, very near start] {3} node[edgelabel, very near end] {2} (m_3);

\draw [] (j_3) -- node[edgelabel, very near start] {2} node[edgelabel, very near end] {2} (m_1);
\draw [] (j_3) -- node[edgelabel, very near start] {3} node[edgelabel, very near end] {3} (m_2);
\draw [ultra thick] (j_3) -- node[edgelabel, very near start] {1} node[edgelabel, very near end] {3} (m_3);

\end{tikzpicture}
\end{center}
\caption{A cycle of best response blocking edges}
\label{fig:cycle}
\end{figure}
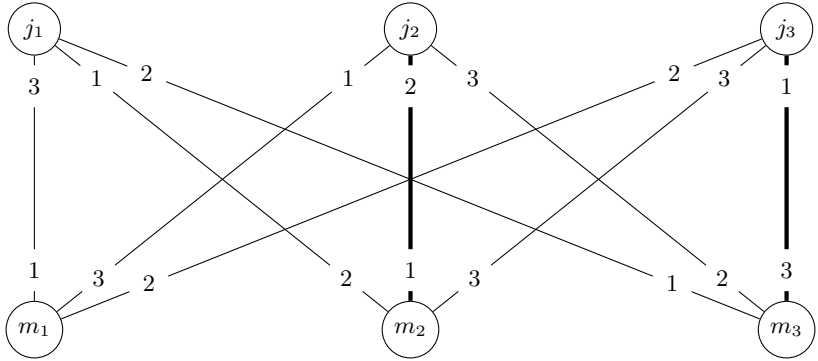

Besides these works on the classical stable marriage problem, there is a number of papers investigating variants of it from the paths-to-stability point of view. For the stable roommates problem, the non-bipartite version of the stable marriage problem, it is known that there is a series of myopic operations that leads to a stable solution, if one exists~\cite{Diamantoudi200418}. A path to stability also exists in the bipartite matching case with payments where flexible salaries and productivity are taken into account~\cite{chen_decentr}. In the hospitals/residents assignment problem, when couples are present, the existence of such a path is only guaranteed if the preferences are weakly responsive~\cite{klaus_pts}. Weak responsiveness ensures consistence between the preferences of each partner and the couple's preference list on pairs of hospitals. In many-to-many markets, supposing substitutable preferences on one side and responsive preferences on the other side, a path to stability can be found~\cite{unver_kojima}. Both substitutable and responsive preferences are defined in instances where preferences are given on sets of vertices. Although many variants of the stable marriage problem have been studied, no paper discusses the case of allocations (instead of matchings or $b$-matchings), where edges are capacitated, thus, they can be partially in stable solutions. Our present work makes an attempt to fill this gap in the literature.

\paragraph{Structure of the paper} In the next section, the essential theoretical basis is provided: stable allocations, and better and best response modifications on such instances are defined. In Section~\ref{sec:corr}, a special case of allocation instances are investigated. We show that although random best response processes generally run in exponential time, in the case of correlated markets, polynomial convergence is expected. Better and best response dynamics in the general case on rational input are extensively studied in Section~\ref{sec:rat}. We describe two deterministic algorithms that generalize the result of Ackermann et al.\ on one-to-one matching markets to stable allocation instances and also show algorithmic differences between the two strategies. In the case of random procedures, convergence is shown for both strategies. Section~\ref{sec:irrat} focuses on running time efficiency. There, a better response algorithm is presented that terminates with a stable solution in $O(|V|^2|E|)$ time, even for irrational input data. A counterexample proves that such an acceleration for the best response dynamics cannot be reached. 

\begin{table}[h]
	\centering
		\begin{tabular}{|c|c c|}
		\hline
			\phantom{n}& shortest path to stability & random path to stability\\ \hline
			best response dynamics & exponential length & converges with probability 1\\
			better response dynamics & polynomial length & converges with probability 1\\
		\hline
		\end{tabular}
\caption{Our results for rational input}
\end{table}

Applied to a matching instance, our best-response algorithm performs the same steps as the two-phase best response algorithm of Ackermann et al. Our better-response variant can also be interpreted as an extended version of the above mentioned method. The only difference is that while our first phase is better response, theirs is best response. However, this seems to be a minor difference, as their proof is also valid for a better response first phase, and our proof still holds if only best blocking edges are chosen. Moreover, stable allocations might be the most complex model in which this approach brings results. The most intuitive extension of Ackermann's algorithm for stable flows~\cite{egres-09-11} does not even result in feasible myopic changes. 

On the other hand, our accelerated better-response algorithm generalizes another known method. Applied directly to the instance with the empty allocation, the accelerated Phase~II performs augmentations like the augmenting path algorithms of Ba\"iou and Balinski, and of Dean and Munshi. Since our algorithm is an accelerated version of our first algorithm, our concept offers a bridge between two known methods for solving two completely different problems, providing a solution to both of them.

\section{Preliminaries}
\label{sec_pre}

\subsection{Stable allocations}

The marriage problem has been extended in several directions. A great deal of research effort has been spent on \emph{many-to-one} and \emph{many-to-many matchings}, sometimes also referred to as $b$-matchings. Their extension is called the \emph{stable allocation problem}, also known as the ordinal transportation problem, since it is a
direct analog of the classical cost-based transportation problem. In this problem, the vertices of a bipartite graph $G=(V,E)$ have \emph{quotas} $q: V \rightarrow \mathbb{R}_{\geq 0}$, while edges have \emph{capacities} $c: E \rightarrow \mathbb{R}_{\geq 0}$. Both functions are \emph{real-valued}, unlike the respective functions in many-to-many instances, where capacities are unit, while quotas are integer-valued. Therefore, allocations can model more complex problems, for example where goods can be divided unequally between agents. In order to avoid confusion caused by terms associated with the marriage model, we call the vertices of the first color class \emph{jobs} and the remaining vertices \emph{machines}. For each machine, its quota is the maximal time spent working. A job's quota is the total time that machines must spend on the job in order to complete it. In addition, machines have a limit on the time spent on a specific job; this is modeled by edge capacities. A feasible allocation is a set of contracts where no machine is overwhelmed and no job is worked on after it has been completed.

\begin{definition} [allocation]
	Function $x: E \rightarrow \mathbb{R}_{\geq 0}$ is called an \emph{allocation} if both of the following hold for every edge $e \in E$ and every vertex $v \in V$ of $G$:
	\begin{enumerate}
		\item $x(e) \leq c(e)$;
		\item $x(v) := \sum_{e \in \delta(v)} x(e) \leq q(v),$ \text{ where $\delta(v)$ is the set of edges incident to~$v$.}
	\end{enumerate}
\end{definition}

To define stability we need \emph{preference lists} as well. All vertices rank their incident edges strictly. Vertex $v$ prefers $uv$ to~$wv$, if $uv$ has a lower rank on $v$'s preference list than $wv$: $\rank_v(uv) < \rank_v(wv)$. In this case we say that $uv$ \emph{dominates} $wv$ at~$v$. A stable allocation instance consists of four elements: $\mathcal{I} = (G, q, c, O)$, where $O$ is the set of all preference lists.

\begin{definition}[blocking edge, stable allocation]
\label{def_st_all}
	An allocation $x$ is \emph{blocked} by an edge $jm$ if all of the following properties hold:
	\begin{enumerate}
		\item $x(jm) < c(jm)$;
		\item\label{j_dom}  $x(j) < q(j)$ or $j$ prefers $jm$ to its worst edge with positive value in $x$;
		\item $x(m) < q(m)$ or $m$ prefers $jm$ to its worst edge with positive value in $x$.
	\end{enumerate}
	A feasible allocation is \emph{stable} if no edge blocks it.
\end{definition}

In other words, edge $jm$ is blocking if it is unsaturated and neither end vertices of $jm$ could fill up its quota with at least as good edges as~$jm$. If an unsaturated edge fulfills the second criterion, then we say that it \emph{dominates} $x$ at~$j$. Similarly, if the third criterion is fulfilled, then we talk about an edge dominating $x$ at~$m$.

\subsubsection{Example for a stable allocation instance.}
The figure below illustrates a stable allocation instance. We use the same example throughout the entire paper to demonstrate different notions defined here. For the sake of simplicity, all edge capacities are unit. The numbers over and under the vertices represent the quota function. The preferences can be seen on the edges: the more preferred edges carry a higher rank, a smaller number. For example, machine $m_1$'s most preferred job is $j_2$, its second choice is $j_3$, while its least preferred, but still acceptable job is~$j_1$. The function $x = 1$ on the thick edges and $x = 0$ on the remaining edges is a feasible allocation, since no quota or capacity constraint is harmed. The unique blocking edge is easy to find: $j_3m_1$ blocks~$x$, because it is unsaturated and both end vertices have some free quota.

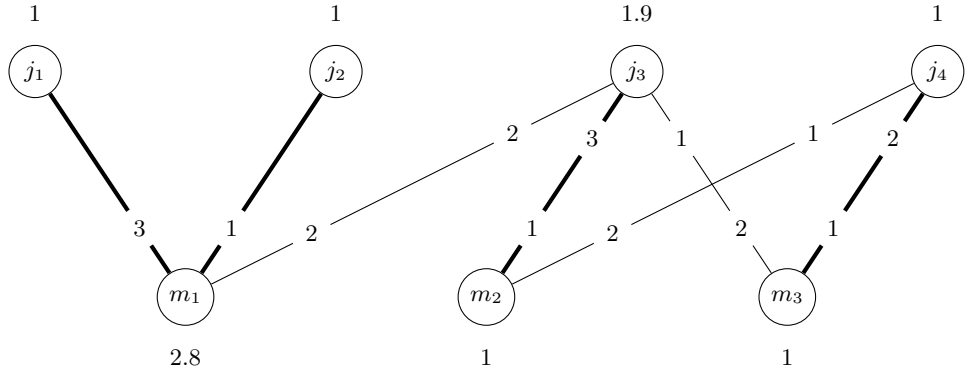
\begin{figure}[H]
\begin{center}
\begin{tikzpicture}[scale=1, transform shape]

\tikzstyle{vertex} = [circle, draw=black]
\tikzstyle{edgelabel} = [circle, fill=white]

\node[vertex] (j_3) at (0, 3) {$j_3$};
\node[above=0.2 cm of j_3] {$1.9$};
\node[vertex] (j_4) at (4, 3) {$j_4$};
\node[above=0.2 cm of j_4] {$1$};
\node[vertex] (j_2) at (-4, 3) {$j_2$};
\node[above=0.2 cm of j_2] {$1$};
\node[vertex] (j_1) at (-8, 3) {$j_1$};
\node[above=0.2 cm of j_1] {$1$};

\node[vertex] (m_1) at (-6,0) {$m_1$};
\node[below=0.2 cm of m_1] {$2.8$};
\node[vertex] (m_2) at (-2, 0) {$m_2$};
\node[below=0.2 cm of m_2] {$1$};
\node[vertex] (m_3) at (2,0) {$m_3$};
\node[below=0.2 cm of m_3] {$1$};

\draw [] (j_3) -- node[edgelabel, near start] {2} node[edgelabel, near end] {2} (m_1);
\draw [] (j_3) -- node[edgelabel, near start] {1} node[edgelabel, near end] {2} (m_3);
\draw [ultra thick] (j_3) -- node[edgelabel, near start] {3} node[edgelabel, near end] {1} (m_2);
\draw [] (j_4) -- node[edgelabel, near start] {1} node[edgelabel, near end] {2} (m_2);
\draw [ultra thick] (j_4) -- node[edgelabel, near start] {2} node[edgelabel, near end] {1} (m_3);

\draw [ultra thick] (j_2) -- node[edgelabel, near end] {1} (m_1);
\draw [ultra thick] (j_1) -- node[edgelabel, near end] {3} (m_1);
	
\end{tikzpicture}
\end{center}
\caption{Feasible, but unstable allocation}
\label{fig:ex}
\end{figure}

Ba\"iou and Balinski~\cite{DBLP:journals/mor/BaiouB02} prove that stable allocations always exist. They also give two algorithms for finding them, an extended version of the Gale-Shapley algorithm and an inductive algorithm. The worst case running time of the first algorithm is exponential, but the latter one runs in strongly polynomial time. Dean and Munshi~\cite{DBLP:journals/algorithmica/DeanM10} speed up the polynomial algorithm using sophisticated data structures: their version runs in $O(|E| \log |V|)$ time for any real-valued instance.

\subsection{Better and best response steps for allocations}

First, we provide some basic definitions and notations we will use throughout the entire paper. A feasible, but possibly unstable allocation $x$ is given at the beginning, the instance can be written as $\mathcal{I}$$ = (G, q, c, O, x)$. Increasing $x$ along a blocking edge and possibly decreasing it along worse edges is a better response step: through this operation, both end vertices of the blocking edge come better off. The definition of better and best response strategies is not as straightforward as it is in the matching instance with unit quotas and capacities. Here, the possible outcomes of a player are ordered lexicographically.

Although lexicographical order seems to be a natural choice, it is somewhat against the convention when discussing stable allocations. In most cases, when comparing the position of an agent in two stable allocations, the so called \emph{min-min criterion} is used~\cite{DBLP:journals/mor/BaiouB02}. According to this rule, the agent prefers the allocation in which its worst positive edge is ranked higher. In order to make use of such an ordering relation, each vertex has to have the same allocation value in all stable solutions. Therefore here, when studying and comparing arbitrary feasible allocations, this concept proves to be counter-intuitive.

In our instance $\mathcal{I}$, jobs form the active side $J$, while machines $M$ are passive players. For sake of simplicity we denote the residual capacity $c(jm)-x(jm)$ of edge $jm$ by $\bar{x}(jm)$ and similarly, the residual quota $q(v)-x(v)$ of vertex $v$ by~$\bar{x}(v)$.

An active player $j$ having some blocking edges is chosen to perform a \emph{best response step} on the current allocation~$x$. Amongst $j$'s blocking edges, let $jm$ be the one ranked highest on $j$'s preference list. The aim of player $j$ is to reach its best possible lexicographical position via increasing~$x(jm)$. To this end, $j$ is ready to allocate all its remaining quota $\bar{x}(j)$ to~$jm$, moreover, it reassigns allocation from all edges worse than $jm$ to~$jm$. Thus, $j$ aims to increase $x(jm)$ by $\bar{x}(j) + x(\text{edges dominated by }jm \text{ at }j)$. To preserve feasibility, $x(jm)$ is not increased by more than~$\bar{x}(jm)$. The passive player $m$ agrees to increase $x(jm)$ as long as it does not lose allocation on better edges. This constraint gives the third upper bound, $\bar{x}(m) + x(\text{edges dominated by }jm \text{ at }m)$. To summarize this, in a best response step $x(jm)$ is increased by the following amount.
\begin{align*}
A := \min\{\bar{x}(j) + x(\text{edges dominated by }jm \text{ at }j), \bar{x}(jm), \\ \bar{x}(m) + x(\text{edges dominated by }jm \text{ at }m)\}
\end{align*}
Once this $A$ and the new $x(jm)$ is determined, $j$ and $m$ fill their remaining quota, then refuse allocation on their worst allocated edges, until $x$ becomes feasible.

\emph{Better response steps} are much less complicated to describe. The chosen active vertex $j$ increases the allocation on an arbitrary blocking edge~$jm$. Both $j$ and $m$ are allowed to refuse allocation on worse edges than~$jm$. This rule guarantees that $j$'s lexicographical situation develops and that the change is myopic for both vertices. By definition, best response steps are always better response steps at the same time. The execution of a single better response step consists of modifications on at most $|\delta(j)| + |\delta(m)| - 1 \leq |V|-1$ edges.

In our example above, $j_3$ and $m_1$ mutually agree to allocate value 1 to $j_3m_1$. If best response strategies are played, $m_1$ refuses 0.2 amount of allocation from~$j_1$, while $j_3$ reduces $x(j_3m_2)$ to 0.9. Through this step, they induce blocking somewhere else in~$G$: now $j_4m_2$ blocks the new~$x$, because $m_2$ lost some allocation. Thus, another myopic change would be to increase $x(j_4m_2)$, and so on. A better response step of the same vertex $j_3$ would be for example to increase $x(j_3 m_1)$ to 1, while refusing $j_3m_2$ entirely. To keep feasibility, $m_1$ has to refuse 0.2 amount of allocation from~$j_1$.

\section{Correlated markets}
\label{sec:corr}

Before tackling the general paths to stability problem, we first restrict ourselves to instances with special preference profiles. In this section, we study the case of stable allocations on an uncoordinated market with correlated preferences. Later we will prove that the convergence time of random best and better response strategies is exponential on general instances. By contrast, here we show that on correlated markets, random best response strategies terminate in expected polynomial time, even in the presence of irrational data. At the end of this section we also elaborate on the behavior of better response dynamics.

\begin{definition}[correlated market]
An allocation instance is \emph{correlated}, if there is a function  $f: E \rightarrow \mathbb{N}$ such that $\rank_v(uv) < \rank_v(wv)$ if $f(uv) < f(wv)$ for every $u, v, w \in V$ and no two edges have the same $f$ value.
\end{definition}

Correlated markets are also called \emph{instances with globally ranked pairs} or \emph{acyclic markets}. The latter property means that there is no cycle of edges such that every edge is preferred to the previous one by their common vertex. Abraham et al.~\cite{enlighten4494} show that acyclic markets are correlated and vice versa. The graph pictured on Figure~\ref{fig:ex} is not correlated: edges $(j_3m_3, j_4m_3, j_4m_2, j_3m_2)$ build a preference cycle.
Ackermann et al.~\cite{ackermann2011uncoordinated} were the first to prove that random better and best response dynamics reach a stable matching on correlated markets in expected polynomial time. Using a similar argumentation, we extend their result to allocation instances.

\begin{theorem}
\label{th:cor}
On correlated allocation instances with real-valued input data, random best response dynamics reach a stable solution in expected time~$O(|V|^2 |E|)$.
\end{theorem}

\begin{proof}
	Before studying paths to stability we show that on correlated instances, the set of solutions has cardinality one. There is an absolute minimum of~$f(jm)$. The single edge $jm$ with this minimal $f$ value must be in all stable allocations with value $\min{\left\{c(jm), q(j), q(m)\right\}}$, otherwise it is blocking. Fixing $x$ on $jm$ and decreasing the quotas of $j$ and $m$ respectively leads to another correlated allocation instance. On this instance, the stable solutions are exactly the stable solutions of the original instance without~$jm$. This leads to an inductive algorithm that proves that there is a unique stable allocation on correlated markets. We will show that random best response dynamics reach this unique solution in expected polynomial time.

	Whenever a job $j$ with an unsaturated edge $jm$ of an absolute minimal $f(jm)$ is chosen to submit an offer, its best response strategy is to increase $x$ on~$jm$. Due to this single best response operation performed by $j$, $x(jm)=\min{\left\{c(jm), q(j), q(m)\right\}}$ is reached. The probability that a vertex $j \in J$ is chosen to take the next step is at least~$\frac{1}{|J|}$. As mentioned above, one best response step requires at most $O(|V|)$ modifications. Thus, in order to reach $x(jm)=\min{\left\{c(jm), q(j), q(m)\right\}}$ on the best edge in~$G$, $|J| \cdot |V| \sim O(|V|^2)$ modifications are needed in expected time. After this $jm$ with minimal $f$ value reached its final position in the unique stable allocation, $x(jm)$ will never be reduced, because neither~$j$, nor $m$ have a better neighboring edge. Thus, $x(jm)$ can be fixed, and a new minimum of $f$ can be chosen for the same procedure as before. The number of iterations is bounded from above by the number of edges in the graph. The unique stable allocation is reached this way in $O(|V|^2 |E|)$ time in expectation.\end{proof}

In order to establish a similar result for better response dynamics in real-valued instances, an exact interpretation of random events would be needed. In the matching case, best and better response dynamics differ exclusively in the rank of the chosen blocking edge: when playing best response strategy, the best blocking edge is chosen by an active vertex~$j$. In contrast to this, here, better response steps differ also in the amount of modification and in the edges chosen to refuse allocation along. The first factor indicates a continuous sample space.

If we assume that any better response step results in reassigning the highest possible allocation value to an arbitrary blocking edge, an analogous proof can be derived. The only difference is that after $j$ is chosen, the expected time of reaching $x(jm)=\min{\left\{c(jm), q(j), q(m)\right\}}$ is larger. In this case, $j$ chooses $jm$ with probability at least~$\frac{1}{|\delta(j)|}$. This implies that reaching the stable allocation value on the best edge takes $|\delta(j)| \cdot (|\delta(j)|+|\delta(m)|-1) \sim O(|V|^2)$ steps in expectation. In total, for all vertices $j \in J$ and all edges the algorithm takes $O(|V|^3 |E|)$ steps in expectation.


\section{Best and better responses with rational data}
\label{sec:rat}

In this section, the case of allocations on an uncoordinated market \emph{with rational data} is studied. As already mentioned, better and best response dynamics can cycle in such instances. We describe two deterministic methods, a better-response and a best-response algorithm that yield stable allocations in finite time. The main idea of our algorithms is to distinguish between blocking edges based on the type of blocking at the job: dominance or free quota.

A blocking edge can be of two types. Recall point~\ref{j_dom} of Definition~\ref{def_st_all}: if $jm$ blocks~$x$, then $x(j) < q(j)$ or $j$ prefers $jm$ to its worst edge with positive value in $x$. We talk about \emph{blocking of type~I} in the latter case, if $jm$ blocks $x$ because $j$ prefers $jm$ to its worst edge having positive value in~$x$. \emph{Blocking of type~II} means that $j$ has no allocated edge worse than~$jm$, but $j$ has not filled up its quota yet,~$x(j) < q(j)$. Note that the reason of the blocking property at $m$ is not involved when defining the two groups.

\subsection{Better response dynamics}


First, we provide a deterministic algorithm that constructs a finite path to stability from any feasible allocation. In the first phase of our algorithm, only blocking edges of type~I are chosen to perform myopic changes along. The active vertices (jobs) choose one of their blocking edges of type~I, not necessarily the best one. In all cases, withdrawal is executed along worst allocated edges. The amount of allocation set to the better edge is determined in such a way that at least one edge or a vertex becomes saturated or empty. Active vertices replace their worst edges with better ones, even if they had free quota. When no blocking edge of type~I remains, the second phase starts. The allocation value is increased on blocking edges of type~II such that they cease to be blocking. The runtime of our algorithm is exponential. Later, in Section~\ref{sec:irrat} we will also show that this algorithm can be accelerated such that a stable solution is reached in polynomial time. The detailed proof of correctness, a pseudocode and execution on a sample instance are also provided.

\begin{theorem}
\label{th:better_rat}
For every allocation instance with rational data and a given feasible allocation~$x$, there is a finite sequence of better responses that leads to a stable allocation.
\end{theorem}

The main idea of the proof is the following. We need to keep track of the change in total allocation value and in the lexicographical position of the active vertices simultaneously. In one step of the first phase along edge~$jm$, either both $j$ and $m$ refuse edges, thus, the allocation value $|x| = \sum_{j \in J} x(j)$ decreases, or only $j$ does so, keeping $|x|$ and improving its situation lexicographically. Since both procedures are monotone and the second one does not impair the first one, the first phase terminates. Termination for the second phase is implied by the fact that passive vertices improve their lexicographical situation in each step. 

\begin{proof}

Recall our example on Figure~\ref{fig:ex}. The unique blocking edge $j_3m_1$ is of type~I, because $j_3$, its active vertex prefers edge $j_3m_1$ to its worst allocated edge~$j_3m_2$.

In the first phase, the jobs propose along \textit{arbitrary} blocking edges of type~I. We will show that this process ends with an allocation where no job has a blocking edge of type~I. In the second phase, the jobs propose along their \textit{best} blocking edges of type~II. Later we will see that during this phase until termination, no job gets a blocking edge of type~I. A pseudocode is provided after the description of both phases.

\textbf{First phase. }
In one step, an arbitrary blocking edge $jm$ of type~I is chosen. Both end vertices, $j$ and $m$ may refuse some allocation along edges when increasing $x$ on~$jm$. Job~$j$ has a \emph{refusal pointer} $r(j)$ that denotes the worst allocated edge to~$j$, if any exists. Similarly, $r(m)$ stands for the worst currently allocated edge of~$m$. A step of Phase~I consists of two or three operations, each along $jm, r(j)$ and possibly along~$r(m)$. Two operations take place, if $m$ has not filled up its quota yet. In this case, $x(r(j))$ will be decreased by $A=\min{\left\{x(r(j)), \bar{x}(jm), \bar{x}(m)\right\}}$. At the same time, $x(jm)$ is increased by the same amount. Depending on which expression is the minimal one, edge $r(j)$ becomes empty or $jm$ becomes saturated or $m$ fills up its quota. Note that $r(m)$ plays no role, because $m$ does not aim to refuse any allocation. In the remaining case, if $m$ has a full quota, three operations take place, since $m$ has to refuse some allocation. The amount of allocation we deal with is now $A=\min{\left\{x(r(j)), \bar{x}(jm), x(r(m))\right\}}$. The allocation on the blocking edge $jm$ will be increased by this~$A$, on the other two edges it will be decreased by~$A$, until one of them becomes empty or saturated. We emphasize that whenever a job $j$ with free quota adds a new edge better than its worst allocated edge to~$x$, it withdraws some allocation from the worst edge. This is what we earlier phrased as freezing quotas: in Phase~I, $j$ acts as if $q(j)$ would be the current allocation value $x(j)$.

We return to our example again. It has already been mentioned that the unique blocking edge $j_3m_1$ is of type~I. The refusal pointer $r(j_3)$ is~$j_3m_2$. Since $m_1$ has not filled up its quota yet, its refusal pointer $j_1m_1$ is irrelevant at the moment. Due to the same reason, two operations take place. The amount we augment with is $\min{\{x(j_3m_2), \bar{x}(j_3m_1), \bar{x}(m_1)\}} = 0.8$. After this operation, $x(j_3m_1) = 0.8, x(j_3m_2) = 0.2$, and $j_3m_1$ is still a Phase~I blocking edge. Since $x(m_1) = q(m_1)$ holds now, three operations are executed. $A = \min{\{x(j_3m_2), \bar{x}(j_3m_1), x(j_1m_1)\}} = 0.2$. Now $j_3m_1$ is saturated, hence it ceases to be blocking. During the first operation, $j_4m_2$ became blocking of type~I, because $m_2$ lost allocation. In the next step, one unit of allocation is reallocated to $j_4m_2$ from~$j_4m_3$. But then, $j_3m_3$ becomes blocking of type~I, and so on.

We use the following multicriteria potential function in order to show that the process does not cycle:
$$\Theta(x) := \left(\Theta_1(x),\Theta_2(x)\right) := \left(\sum_{j \in J}{x(j)}, \quad \sum_{j \in J}{\sum_{jm \in E}{x(jm) \rank_j(jm)}}\right)$$

Recall that $\rank_j(jm)$ stands for the rank of $jm$ on $j$'s preference list. The smaller $\rank_j(jm)$ is, the better $m$ for $j$ is. In the expression above, both components are non-negative for any feasible allocation~$x$. They both have an upper bound as well: $$0 \leq \Theta_1(x) \leq \sum_{jm \in E} {c(jm)}, \quad 0 \leq \Theta_2(x) \leq |J| \cdot \max_{jm \in E} {c(jm)} \cdot \max_{j \in J}{|\delta(j)|}.$$ We will show that this function lexicographically decreases in each step of the procedure. The process terminates if the amount of increment is always greater than a fixed positive constant. If all data are rational, this is guaranteed.

Considering the potential function, we need to keep track of those two jobs that proposed or got refused, since the contribution of all other jobs remains the same. Thus, their terms in the summations of $\Theta(x)$ do not change.

As mentioned above, a step consists of either two or three edges changing their value in~$x$. In the first case, when only two edges change their value in~$x$, there is only one job~$j$ that modifies its contribution. The allocated value of this vertex remains the same, thus, $\Theta_1(x)$ does not change. But $\Theta_2(x)$ does, because some allocation will move from a less preferred edge to~$jm$. This ensures the decrement of~$\Theta_2(x)$. In the second case, where three edges are involved, there is a job~$j$ that improves its lexicographical position, and another job $j'$ that loses allocation. The effect of the first change at $j$ is just as above, $\Theta_1(x)$ remains the same, $\Theta_2(x)$ decreases. Losing allocation for $j'$ means that both terms decrease, since $x(j')$ decreases.

\textbf{Second phase.}
In the second phase, we work with the original quota function. When developing the allocation along a blocking edge $jm$ of type~II, $m$ may refuse some allocation, but $j$ may not, since the reason of blocking is that $j$ has not filled up its (original) quota yet. Thus, we do not need the pointer $r(j)$ any more. One step consists of changes along one edge if~$x(m) < q(m)$, or along two edges otherwise. If $m$ has not filled up its quota yet, then we simply assign as much allocation to $jm$ as possible without exceeding $q(j), q(m)$ or~$c(jm)$. If $m$ has to refuse something from a job $j'$ in order to accept better offers from~$j$, we improve $m$'s position until $j'm$ becomes empty or $jm$ becomes saturated or $j$ gets its quota filled up.

First, we need to see that no step can induce a blocking edge of type~I. One step in Phase~II leaves all vertices but $j,m$ and the possibly refused $j'$ unchanged. Thus, if there is a blocking edge of type~I after the modification, it must be incident to one of those vertices. The three cases are the following.
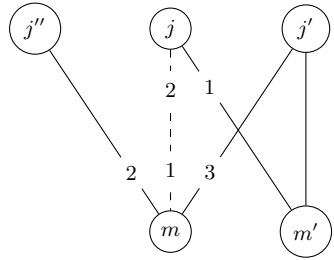
\begin{figure}[H]
\begin{center}
\begin{tikzpicture}[scale=0.9, transform shape]
\tikzstyle{vertex} = [circle, draw=black]
\tikzstyle{edgelabel} = [circle, fill=white]
\node[vertex] (j'') at (0, 3) {$j''$};
\node[vertex] (j) at (2, 3) {$j$};
\node[vertex] (j') at (4, 3) {$j'$};
\node[vertex] (m) at (2, 0) {$m$};
\node[vertex] (m') at (4,0) {$m'$};
\draw [] (j') -- node[edgelabel, near end] {3} (m);
\draw [] (j') -- (m');
\draw [] (j'') -- node[edgelabel, near end] {2} (m);
\draw [dashed] (j) -- node[edgelabel, near start] {2} node[edgelabel, near end] {1} (m);
\draw [] (j) -- node[edgelabel, near start] {1} (m');
\end{tikzpicture}
\end{center}
\caption{Edges affected by one myopic operation}
\end{figure}

\begin{itemize}
	\item $j''m$ blocks~$x$. The position of $m$ became lexicographically better, thus, no new blocking edge incident to $m$ could be introduced. The existing blocking edges $j''m$ of type~II cannot become of type~I, because $j''$'s position remained unchanged.
	\item $jm'$ (or $jm$) blocks~$x$. The only change at $j$ is that $x(jm)$ increases, thus, $j$ also develops its  lexicographical position. From this follows that no new blocking edge incident to $j$ appeared. Blocking edges of type~II change their type of blocking if $j$ increased its allocation on a worse edge. This is the point where we use the rule of choosing the best blocking edge in Phase~II as~$jm$, which contradicts to this assumption.
	\item $j'm'$ (or $j'm$) blocks~$x$. The only change in $j'$'s neighborhood is that $x(j'm)$ decreases. After this step, consider an unsaturated edge $j'm'$ preferred by $j'$ to its worst allocated edge. Since no machine worsens its lexicographical position in Phase~II, if $j'm'$ dominates the new allocation~$x$, it already dominated the previous allocation. Thus, $j'm'$ must have been a blocking edge prior to the modification.  
\end{itemize}

From this we see that once Phase~II has started, Phase~I can never return. The last step ahead of us is to show that Phase~II may not cycle. But this follows from the fact that in each step exactly one machine strictly improves its lexicographical situation, while all other machines maintain the same allocation as before. In case of a rational input, this improvement is bounded from below, thus, the second phase of the algorithm terminates.

\begin{algorithm}[H]
\renewcommand{\thealgorithm}{} 
\caption{Two-phase better response algorithm}
\label{alg:two_ph}
\begin{algorithmic}
\While{$\exists j\in J$ with a blocking edge of type~I}
	\State \textsc{Improvement I}($j$)
\EndWhile
\While{$\exists j\in J$ with a blocking edge of type~II}
	\State \textsc{Improvement II}($j$)
\EndWhile
\end{algorithmic}
\end{algorithm}

\begin{algorithmic}
\begin{figure}[H]
\centering
\begin{tabular}[H]{cc}
\begin{minipage}[t]{3.0 in}
\Procedure{Improvement I}{$j$}
	\State $jm \leftarrow$ blocking edge of type~I of $j$
	\If{$x(m) < q(m)$}
		\State  $A := \min{\left\{x(r(j)),\bar{x}(jm), \bar{x}(m)\right\}}$
		\State  $x(r(j)) := x(r(j))-A$
		\State  $x(jm) := x(jm)+A$
	\Else
		\State  $A := \min{\left\{x(r(j)),\bar{x}(jm), x(r(m))\right\}}$
		\State  $x(r(j)) := x(r(j))-A$
		\State  $x(jm) := x(jm)+A$
		\State  $x(r(m)) := x(r(m))-A$
	\EndIf
\EndProcedure
\end{minipage}&

\hfill
\begin{minipage}[t]{3.0 in}
\Procedure{Improvement II}{$j$}
	\State $jm \leftarrow$ best blocking edge of type~II of $j$
	\If{$x(m) < q(m)$}
		\State  $A := \min{\left\{\bar{x}(jm), \bar{x}(j), \bar{x}(m)\right\}}$
		\State  $x(jm) := x(jm)+A$		
	\Else
		\State  $A := \min{\left\{x(r(m)),\bar{x}(jm), \bar{x}(j)\right\}}$
		\State  $x(jm) := x(jm)+A$
		\State  $x(r(m)) := x(r(m))-A$
	\EndIf
\EndProcedure
\end{minipage}
\end{tabular}
\end{figure}
\end{algorithmic}

\end{proof}

The duration of both phases strongly depends on the capacities and quotas. The examples below show two bad instances. The capacity is $N$ on all edges, where $N$ is an arbitrarily big integer. Quotas are marked above and below the vertices. The initial allocation for Phase~I is $N$ on $j_1m_1$ and on $j_2m_2$ and zero on the remaining two edges. The first phase performs $N$ augmenting steps along the same cycle. Phase~II terminates after $N$ iterations in the second instance, starting with the empty allocation.
\begin{figure}[H]
\begin{center}
\begin{tikzpicture}[scale=0.9, transform shape]
\tikzstyle{vertex} = [circle, draw=black]
\tikzstyle{edgelabel} = [circle, fill=white]
\node[vertex] (j_1) at (0, 3) {$j_1$};
\node[above=0.2 cm of j_1] {$N$};
\node[vertex] (j_2) at (4, 3) {$j_2$};
\node[above=0.2 cm of j_2] {$N$};
\node[vertex] (m_1) at (0, 0) {$m_1$};
\node[below=0.2 cm of m_1] {$N$};
\node[vertex] (m_2) at (4,0) {$m_2$};
\node[below=0.2 cm of m_2] {$N+1$};
\draw [] (j_1) -- node[edgelabel, near start] {1} node[edgelabel, near end] {2} (m_2);
\draw [ultra thick] (j_1) -- node[edgelabel, near start] {2} node[edgelabel, near end] {1} (m_1);
\draw [] (j_2) -- node[edgelabel, near start] {1} node[edgelabel, near end] {2} (m_1);
\draw [ultra thick] (j_2) -- node[edgelabel, near start] {2} node[edgelabel, near end] {1} (m_2);
\node[vertex] (j'_1) at (8, 3) {$j_1$};
\node[above=0.2 cm of j'_1] {$N+1$};
\node[vertex] (j'_2) at (12, 3) {$j_2$};
\node[above=0.2 cm of j'_2] {$N$};
\node[vertex] (m'_1) at (8, 0) {$m_1$};
\node[below=0.2 cm of m'_1] {$N$};
\node[vertex] (m'_2) at (12,0) {$m_2$};
\node[below=0.2 cm of m'_2] {$N$};
\draw [] (j'_1) -- node[edgelabel, near start] {1} node[edgelabel, near end] {2} (m'_2);
\draw [] (j'_1) -- node[edgelabel, near start] {2} node[edgelabel, near end] {1} (m'_1);
\draw [] (j'_2) -- node[edgelabel, near start] {1} node[edgelabel, near end] {2} (m'_1);
\draw [] (j'_2) -- node[edgelabel, near start] {2} node[edgelabel, near end] {1} (m'_2);
\end{tikzpicture}
\end{center}
\caption{Worst-case instances}
\label{fig:worst_case}
\end{figure}
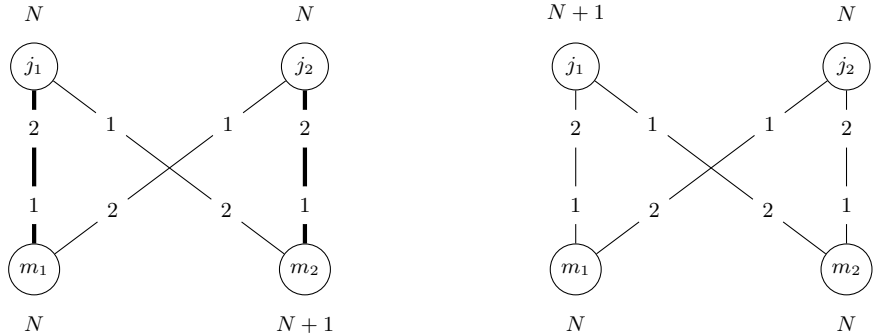

This algorithm also proves an important result regarding rational random better response processes. If the input is rational (there is a smallest positive number that can be represented as a linear combination of all data), it is clearly worthwhile to restrict the set of feasible better response modifications to the ones that reassign a multiple of this unit. For this reason, the set of reachable allocations is finite and they can be seen as states of a discrete time Markov chain. Our algorithm proves that from any state there is a finite path to an absorbing state with positive probability.

\begin{theorem}
\label{th_random_better}
	In the rational case, random better response strategies terminate with a stable allocation with probability one.
\end{theorem}

Polynomial time convergence cannot be shown, since better response strategies need exponential time to converge even in matching instances~\cite{ackermann2011uncoordinated}.

\subsection{Best response dynamics}

In this subsection, we derive analogous results for best response modifications to the ones established for better response strategies. The main difference from the algorithmic point of view is that instances can be found in which no series of best response strategies terminate with a stable solution in polynomial time. A small example resembles the instance given by Ba\"iou and Balinski~\cite{DBLP:journals/mor/BaiouB02} to prove that the Gale-Shapley algorithm requires exponential time to terminate in stable allocation instances. Let $G$ be a complete bipartite graph on four vertices, with quota $q(j_1)=N+1, q(j_2)=q(m_1)=q(m_2)=N$ and initial allocation $x(j_1m_1) = x(j_2m_2) = N$ for an arbitrary large number~$N$. If the preference profile is chosen to be cyclic, such that $\rank_{j_1}(m_1) = \rank_{j_2}(m_2) = \rank_{m_1}(j_2) = \rank_{m_2}(j_1) = 1$, the unique series of best-response steps consists of $2N$ operations. A path of exponential length to stability can still be found. 

\begin{theorem}
\label{th:best_rat}
	For every allocation instance with rational data and a given feasible allocation~$x$, there is a finite sequence of best responses that leads to a stable allocation.
\end{theorem}

\begin{proof}

	Similarly to our method for better response strategies, here we prove that there is a two-phase algorithm that terminates with a stable solution. 
	
	All blocking edges we take into account are best blocking edges of their job~$j$. Depending on their rank compared to $j$'s worst allocated edge $r(j)$, they are either of type~I or type~II. A job $j$'s best blocking edge $jm$ is \begin{itemize}
	\item of type~I(a), if $\rank_j(jm) < \rank_j r(j)$ and \\ $\bar{x}(j) < \min{\{\bar{x}(jm), \bar{x}(m) + x(\text{edges dominated by $jm$ at $m$})\}}$;
	\item of type~I(b), if $\rank_j(jm) < \rank_j r(j)$ and \\ $\bar{x}(j) \geq \min{\left\{\bar{x}(jm), \bar{x}(m) + x(\text{edges dominated by $jm$ at $m$})\right\}}$;
	\item of type~II, if $\rank_j(jm) \geq \rank_j r(j)$.
	\end{itemize}
	The intuitive interpretation of the grouping above is given by the steps needed to execute when $jm$ is chosen to perform a best response operation along. If $jm$ is of type~I(a), then $jm$ can be saturated without any refusal made by~$j$, since $j$ has sufficient free quota. On the other hand, if $j$ agrees to reduce $x(r(j))$ in order to accommodate more allocation on $jm$, then $jm$ is a blocking edge of type~I(b). The remaining case occurs when $jm$ is worse than $r(j)$, that is, $j$ accepts as much allocation from $m$ as much free quota it has. In this case, no rejection is called by~$j$.
	
	In Phase~I, only best blocking edges of type~I(a) and of I(b) are selected. Then, when only type~II blocking edges remain, Phase~II starts. In order to prove finite termination, we introduce two potential functions, $\Theta(x)$ and~$\Psi$. When proving the termination of the first phase, both of them are used, while the second phase is discussed by analyzing the behavior of~$\Psi$. The lexicographical position of machine $m$ is denoted by~\lex$(m)$.
	$$\Theta(x) := (\Theta_{1}(x), \Theta_{2}(x)) := (\rank_j(r(j)), x(r(j)))$$ 
	$$\Psi(x) := - (\text{\lex}(m))$$
	
	\begin{claim}
		The best response step of job $j$ along edge $jm$ of type~I(a) decreases~$\Theta(x)$.
	\end{claim}
	
	\begin{proof}
		Due to the type-defining characteristics listed above, there is a rejection on~$r(j)$. At the same time, $j$ cannot gain allocation on any edge worse than $r(j)$, because the only edge on which $x$ is increased is~$jm$. 
	\end{proof}
	
	\begin{claim}
		The best response step of job $j$ along edge $jm$ of type~I(b) decreases $\Psi(x)$ and does not increase~$\Theta(x)$.
	\end{claim}
	
	\begin{proof}
		 Since $j$ does not reject any allocation, $x(r(j))$ remains unchanged. For the same reason, no machine loses allocation. The only machine whose position changes is $m$ itself: it clearly develops its lexicographical position.
	\end{proof}
	
	For any rational input data, the change in $\Theta(x)$ or $\Psi(x)$ is each round is bounded from below. Since both functions have an absolute minimum, phase~I terminates in finite time.
	
	\begin{claim}
		The best response step of job $j$ along edge $jm$ of type~II decreases~$\Psi(x)$. Moreover, no edge becomes blocking of type~I(a) or~I(b).
	\end{claim}
	
	\begin{proof}
		 During the second phase, no machine loses allocation, thus, their lexicographical position cannot worsen. In addition, \lex$(m)$ develops. This also implies that no edge $j'm'$ dominates $x$ at $m'$ that has not already dominated it before the myopic change. Moreover, edges that lost allocation during that step are the worst-choice edges of $j$, hence they cannot be blocking of type~I(a) or~I(b). If there is an edge $j'm'$ that became blocking of type~I(a) or~I(b), then it is better than the worst edge in $x$ at~$j'$. These edges were already unsaturated before the last step and also already dominated $x$ at both end vertices. It contradicts the fact that best blocking edges are chosen in each step.
	\end{proof}
\end{proof}

The same arguments as above, in Theorem~\ref{th_random_better}, using finite Markov chains imply the result on random procedures.

\begin{theorem}
	In the rational case, random best response strategies terminate with a stable allocation with probability one.
\end{theorem}

\section{Irrational data - a strongly polynomial algorithm}
\label{sec:irrat}

In our previous section, we relied several times on the fact that in each step, $x$ is changed with values greater than a specific positive lower bound. When irrational data are present, \eg, $q, c$ or $x$ are real-valued functions, this cannot be guaranteed. Hence, our arguments for termination are not any more valid. Moreover, both of our algorithms require exponentially many steps to terminate. In this section, we describe a fast version of our two-phase better response algorithm that terminates in polynomial time with a stable allocation also for irrational input data. We also give a detailed proof of correctness for the first phase and show a construction with which all Phase~II steps can be interpreted as Phase~I operations on a slightly modified instance.

\subsection{Accelerated first phase}

The algorithm and the proof of its correctness can be outlined the following way. A helper graph is built in order to keep track of edges that may gain or lose some allocation. A potential function is also defined, it stores information about the structure of the helper graph and the degree of instability of the current allocation. In the helper graph we are looking for walks to augment along. The amount of allocation we augment with is specified in such a way that the potential function decreases and the helper graph changes. When using walks instead of proposal-refusal triplets, more than one myopic operation can be executed at a time. Moreover, we also keep track of consequences of locally myopic improvements. For example, we spare running time by avoiding reducing allocation on edges that later become blocking anyway.

First, we elaborate on the structure of the helper graph, define alternating walks and specify the amount augmentation. The method, the proof of correctness, the pseudocode and a sample execution are all described in details here.


\subsubsection*{Helper graph.}
Recall that our real-valued input $\mathcal{I}$ consists of a stable allocation instance $(G, q, c,O)$ and a feasible allocation~$x$. First, we define a helper graph $H(x)$ on the same vertices as~$G$. This graph is dependent on the current allocation $x$ and will be changed whenever we modify~$x$. The edge set of $H(x)$ is partitioned into three disjoint subsets. The first subset $P$ is the set of Phase~I blocking edges. Each job $j$ that has at least one edge with positive $x$ value, also has a worst allocated edge,~$r(j)$. These are the edges jobs tend to reduce $x$ along when a myopic change is made. These \emph{refusal pointers} form~$R$, the second subset of~$E(H(x))$. We also keep track of edges that are currently not of blocking type~I, but later on they may enter set~$P$. This last subset $P'$ consists of edges that may become blocking of type~I after some myopic changes. An edge $jm \notin P$ has to fulfill three criteria in order to belong to~$P'$: \begin{inparaenum}[1)]
    \item $c(jm) > x(jm)$;
    \item $m$ has at least one refusal edge;
    \item $j$ prefers $jm$ to its worst allocated edge~$r(j)$.
  \end{inparaenum}
Such an edge immediately becomes blocking if $m$ loses allocation along one of its refusal edges. Edges in $P'$ are called \emph{possibly blocking edges}, the set $P \cup P'$ forms the set of \emph{proposal edges}. Note that a job $j$ may have several edges in $P$ and~$P'$, but at most one in~$R$. Moreover, if $j$ has a proposal edge in~$H(x)$, it also has an edge in~$R$. Regarding the machines, if $m$ has a $P'$-edge, it also has an $R$-edge. The following lemma provides an additional structural property of~$H(x)$.

\begin{lemma}
\label{p_good}
	If $jm \in P$ and $j'm \in P'$, then $\rank_m(jm) < \rank_m(j'm)$. \\That is,blocking edges are preferred to possibly blocking edges by their common machine~$m$.
\end{lemma}

\begin{proof}
	Since $jm \in P$ is a blocking edge of type~I, $jm$ dominates $x$ at~$m$. If the statement is false, then $\rank_m(jm) > \rank_m(j'm)$ for some unsaturated edge $j'm$ that is better than the worst allocated edge of~$j'$. But then also $j'm$ dominates $x$ at~$m$. This, together with the first and last properties of possibly blocking edges would imply that~$j'm \in P$.
\end{proof}

Once again we return to our example shown on Figure~\ref{fig:ex}. The only blocking edge $j_3m_1$ alone forms~$P$. $R$ contains all four edges with positive allocation value: $j_1m_1, j_2m_1, j_3m_2$ and~$j_4m_3$. Edges $j_3m_3$ and $j_4m_2$ are possibly blocking. Thus, in this case, $G = H(x)$.


\subsubsection*{Alternating walks.}

Our algorithm performs augmentations along alternating walks, so that the allocation value of the refusal edges decreases, while the value of proposal edges increases. This is done in such a way that $R$, $P$ or $P'$ (and thus, $H(x)$) changes. The main idea behind these operations is the same we used in our previous proof: reassigning allocation to blocking edges from worse edges, such that the procedure is monotone. The difference between the two methods is that while our first algorithm tackles a single blocking edge in each step, here we deal with a set of blocking edges (forming the alternating walk) at once. 

When choosing the alternating proposal-refusal walk $W$ to augment along, the following rules have to be observed:
\begin{enumerate}
    \item The first edge $jm_1$ is a $P$-edge.
    \item $P$ and $P'$-edges are added to $W$ together with the refusal edge they are incident with on the active side.
    \item\label{bestp} Machines choose their best $P$ or $P'$-edge.
    \item\label{term} $W$ ends at $m$ if \begin{inparaenum}[1)]
	\item $m$ has no proposal edge or 
	\item its best proposal edge goes to a vertex already visited by~$W$.\end{inparaenum}
\end{enumerate}

As long as there is a blocking edge of type~I, the first edge $jm_1$ of such a walk can always be found. Lemma~\ref{p_good} guarantees that point~\ref{bestp} is not harmed by this~$jm_1$. After taking $r(j)$, all that remains is to continue on best proposal edges of machines and refusal edges of jobs they end at. Since $G$ is a finite set, either of the cases listed in point~\ref{term} will appear. According to these rules, proposal-refusal edge pairs are added to the current path until \begin{inparaenum}[1)]
	\item there is no pair to add or 
	\item the path reaches a vertex already visited. \end{inparaenum} In the first case, $W$ is a path. In the latter case, $W$ is a union of a path and a cycle, connecting at exactly one vertex. This vertex is the last vertex listed on~$W$, where our method halts, observing point~\ref{term}. $W$ can be, of course, a single path or a single cycle as well.
	
	
	Before elaborating on the amount of augmentation, we emphasize that $W$ is just a \emph{subset} of the set of edges whose $x$ value changes during an augmentation step. The goal is to reassign allocation from refusal edges to blocking edges, until a stable solution is derived. Naturally, on an alternating walk, refusal edges lose the same amount of allocation proposal edges gain. But, except if augmentations are performed along a single cycle, there is a single machine $m_1$ that gains allocation in total. In order to preserve feasibility, this machine might have to refuse allocation on edges not belonging to~$W$. The exact amount of these refusals is discussed later, together with the amount of augmentation along~$W$. Since no other vertex gains allocation in an augmentation step, feasibility cannot be harmed elsewhere. Thus, these are the only edges not on $W$ that are modified.
	
	 By contrast, if the augmentation is performed along a single cycle~$C$, refusals only happen on $r(j) \in W \cap R$ edges. Even if the machine $m_1$ that started $C$ has a full quota, it does not need to refuse any allocation, since $x(m_1)$ remains unchanged during the augmentation. Note that executing several local myopic steps greedily, like in our first algorithm, would lead to a different output. Then, $m_1$ would refuse edges, not knowing that it loses allocation later. As a result of that, $m_1$ would go under its quota, and would possibly create new blocking edges. Both strategies are better response, the difference is that our second algorithm keeps track of changes made as a consequence of a myopic operation.	

\subsubsection*{Amount of augmentation.}
Once $W$ is fixed, the amount of allocation $A$ has to be determined to augment with. It must be chosen so that
\begin{inparaenum}[1)]
	\item a feasible allocation is derived and 
	\item at least one refusal edge becomes empty or at least one proposal edge leaves~$P \cup P'$. \end{inparaenum} These points guarantee that $H(x)$ changes. To fulfill these two requirements, the minimum of the following terms is determined.
	\begin{enumerate}
		\item Allocation value on refusal edges along $W$: $x(r(j))$, where $r(j) \in W \cap R$.
		\item Residual capacity on proposal edges along $W$: $\bar{x}(p), \bar{x}(p')$, where $p, p' \in W \cap (P  \cup P')$.	
		\item If $W$ is not a single cycle, $m_1$ may refuse sufficient amount of allocation such that $jm_1$ does not become saturated, but it stops dominating $x$ at~$m_1$. In this case, the residual quota of $m_1$ must be filled up and, in addition, the sum of allocation value on edges worse than $jm_1$ must be refused. With this, $jm_1$ becomes the worst allocated edge of a full machine. Until reaching this point, $jm_1$ may gain $\bar{x}(m_1) + x(\text{edges dominated by }jm_1 \text{ at }m_1)$ amount of allocation in total.
	\end{enumerate}
To summarize this, we augment with $A := \min\{x(r(j)), \bar{x}(p), \bar{x}(p')|  r(j) \in W \cap R, p, p' \in W \cap (P  \cup P')\}$ if $W$ is a cycle, because then the last case with the starting vertex $m_1$ may not occur. Otherwise, the amount of augmentation is $A := \min \{x(r(j)), \bar{x}(p), \bar{x}(p'), \bar{x}(m_1) + x(\text{edges dominated by }jm_1 \text{ at }m_1) | r(j) \in W \cap R, p, p' \in W \cap (P  \cup P')\}$.


The second phase of our method can be interpreted as the execution of the first phase on a modified instance. The modification needed consist of introducing a dummy job and swapping the roles of the active and passive color classes.

In total, the algorithm performs $O(|V||E|)$ rounds, each of them needs $O(|V|)$ time to be computed. Thus, it runs in $O(|V|^2|E|)$ time. For a detailed proof of correctness and runtime computation, see the proof below.

\begin{theorem}
		For every real-valued allocation instance and given feasible allocation, there is a sequence of better responses leading to a stable allocation in $O(|V|^2|E|)$ time.
\end{theorem}

\begin{proof}

\textbf{Potential function. }
We show with the help of the following multicriteria potential function that the procedure is monotone and finite:

$$\Theta(x) := \left(\Theta_1(x),\Theta_2(x)\right) :=  \left(\rank_j(r(j)), - \rank_m(\text{best proposal edge at } m) \right).$$

The first component is a vector of $|J|$ entries. The $i$th of them is associated to the $i$th job of the instance and contains the rank of the worst allocated edge of this job. If there is no allocated edge, $\rank_j(r(j))$ can be interpreted as a large number, for example as~$|M|+1$. The second component is also a vector, containing one element for each machine. This element expresses the rank of the best proposal edge incident to~$m$. In lack of proposal edges, this term can also be interpreted as a large constant, for example as~$|J|+1$. In order to keep both terms decreasing, a minus sign is added to this expression.

Each round of our algorithm consists of two operations that change~$x$: procedure augment and update. First, we prove that the $\Theta(x)$ lexicographically strictly decreases after each augmentation. Both $\Theta_1$ and $\Theta_2$ are vectors: decrement means that at least one element of the vector decreases, while no element increases. For $\Theta_1$, such a decrement takes place if no job receives a worse allocated edge than its current edges, and, in addition, at least one job loses its worst allocated edge entirely. The second component $\Theta_2$ decreases if no machine receives a proposal edge better than all of its current proposal edges and at least one machine loses its best proposal edge. As a second step, we will see that the update operation never increases~$\Theta(x)$.

Since $\Theta(x)$ is a bounded, integer-valued function, any procedure that modifies it monotonically, is finite. Later, we elaborate on the running time of our algorithm.

\textbf{Augmentation.} As mentioned above, the our goal here is to show that each augmentation step decreases~$\Theta(x)$. The amount of allocation we augment with depends only on the extreme points of the $\min$ function. Recall the three points we listed when defining the amount of augmentation.
\begin{enumerate}
	\item $x(r(j))$ \\ The worst allocated edge of $j$ becomes empty, while $x(j)$ remains unchanged, hence $\Theta_1(x)$ decreases.
	\item $\bar{x}(p)$ or $\bar{x}(p')$ \\ If one of the proposal edges reaches its capacity, it stops being blocking. Since it was the best blocking edge of its machine, $\Theta_2(x)$ decreases.
	\item In case of walks: $\bar{x}(m_1) + x(\text{edges dominated by } jm_1 \text{ at } m_1)$ \\ The first blocking edge on $W$, $m_1$'s best proposal edge ceases to dominate $x$ at~$m_1$, hence $\Theta_2(x)$ decreases.
\end{enumerate}

Having shown that $\Theta(x)$ decreases every time when either of the augmenting procedures are called, all that remains to show is that the second operation, updating $H(x)$, does not affect this monotonicity negatively. The vertex set of $H(x)$ is fix. In the following, we study the three subsets of $E(H)$ separately and show that updating them never increases~$\Theta(x)$. 

\subsubsection*{Update $R$.}

\begin{lemma}
\label{r_mon}
	During the accelerated Phase~I algorithm, $r(j)$ moves monotonically on $j$'s preference list, always pointing to a better machine.
\end{lemma}

\begin{proof}
	Suppose that there is a refusal pointer that moved to a worse edge. Since $r(j)$ is always the worst allocated edge of~$j$, this implies that $j$ increased $x$ along an edge worse than any of its allocated edges. Allocation is increased only on proposal edges. Since no proposal edge may be worse than the current refusal pointer, it is impossible to increase allocation on such an edge during Phase~I.
\end{proof}

With this lemma we showed that any operation that shifts a refusal pointer develops our potential function~$\Theta(x)$. From this point on, we consider a setting where all refusal pointers are fix. This also implies that $\Theta_1$ does not change, thus, we only concentrate on~$\Theta_2$.

\subsubsection*{Update $P$.}

If $\Theta_2(x)$ increases, then there is a machine $m$ whose best proposal edge became better. Since the preference lists are fix, this is only possible, if an edge that was not in $P \cup P'$ becomes blocking or possibly blocking, moreover, it becomes the best proposal edge of the machine. If it becomes blocking (and not possibly blocking), then update $P$ adds an edge $jm \notin P \cup P'$ to~$P$. Blocking edges of type~I have to fulfill three criteria, at least one of them was not fulfilled before the augmentation.

\begin{enumerate}
	\item $jm$ became unsaturated
		\begin{itemize}
 			\item One of the two possibilities for an edge to loose allocation occurs when $jm \in W \cap R$. Since $jm$ is already the worst allocated edge of~$j$, it may not become a blocking edge of type~I.
			\item Even if $jm \notin W$, it can lose allocation, but only if $x(jm)$ was reduced by~$m = m_1$, the starting vertex of the alternating walk. It is the only machine that refuses allocation, and it does so only if $A - \bar{x}(m_1) > 0$. When the refusal happens, $x(m_1) = q(m_1)$ and $m_1$ has no worse allocated edge than~$jm_1$. In addition, $m_1$ does not lose any allocation in the current step. Since $q(m_1)$ is full with edges better than~$jm_1$, $jm_1$ is not blocking.
		\end{itemize}
	\item $jm$ became better than the worst allocated edge of $j$\\ Lemma~\ref{r_mon} shows that $j$'s worst allocated edge never becomes worse during Phase~I.
 	\item $jm$ became better than the worst allocated edge of $m$ or $m$ became unsaturated
		\begin{itemize}
			\item In the first case, $m$ increased $x$ along an edge worse than~$jm$. This worse edge was in~$P \cup P'$, hence $jm$ already dominated $x$ at $m$ or $m$ already had a refusal pointer. Thus, changing this property is not sufficient for $jm$ to become blocking.
			\item If $m$ lost some allocation, then it was the last vertex on $W$ and thus, it had a refusal pointer. According to out definition of~$P'$, all unsaturated edges that dominate $x$ at their job and are incident to a refusal pointer at their machine already belonged to~$P'$. They may leave this set now and be added to~$P$, but $P \cup P'$ remains unchanged and thus, the best proposal edge at $m$ as well.
		\end{itemize}
\end{enumerate}

\subsubsection*{Update $P'$.}

$\Theta_2(x)$ may also increase, if an edge $jm \notin P \cup P'$ is added to~$P'$ and it becomes the best proposal edge of~$m$. Just like above, we check the three criteria that have to be fulfilled by edges in~$P'$.

\begin{enumerate}
	\item $jm$ became unsaturated
		\begin{itemize}
 			\item First we consider the case when $jm \in W \cap R$. But then $jm$ is a refusal edge and worst allocated edges of jobs are never in~$P'$.
			\item The other option is that $x(jm)$ was reduced by~$m = m_1$. As mentioned above, the only machine that refuses allocation is~$m_1$, the first vertex on an alternating walk. Even if $jm$ becomes a $P'$-edge, $m_1$ had a better proposal edge at the beginning of the augmentation: the first edge of~$W$.
		\end{itemize}
	\item $jm$ became better than the worst allocated edge of $j$\\ We can again rely on Lemma~\ref{r_mon}.
 	\item $m$ gained a refusal edge \\ We supposed that set $R$ is fix, no refusal pointer moves.
\end{enumerate}

We have investigated the effects of the update operation on $\Theta(x)$ for all three subsets of~$E(H(x))$. Each round of the algorithm consists of the following three steps: finding an alternating walk, augmenting along it and then updating~$H(x)$. The first procedure does not change $\Theta(x)$, the second strictly decreases it, and the last one never increases it. Thus, $\Theta(x)$ changes strictly monotonically in each round. Since $\Theta(x)$ is integer-valued and bounded, our algorithm terminates.
\end{proof}

\subsubsection*{Running time.}

The helper graph $H(x)$ has at most as many edges and vertices as~$G$. In each iteration, $\Theta(x)$ develops. Consider first the case when only $\Theta_2$ changes. The best proposal edge of each machine $m$ can walk along all $\delta(m)$ edges of~$m$. Since the procedure is monotone, $|E|$ such steps can be executed in total. Then, $\Theta_1$ has to develop. Just like $\Theta_2$, $\Theta_1$'s monotone behavior also allows $m$ steps in total. Yet it is not possible that both components need all the $m$ rounds. When a refusal pointer $jm$ switches to a better edge~$jm'$, most of the elements in vector $\Theta_2$ remain unchanged.

Suppose the last augmentation along walk $W$ shifted a single refusal pointer~$jm$. We investigate the change in~$\Theta_2$. Clearly, $\Theta_2$ can be increased, since only lexicographical monotonicity of $\Theta(x)$ can be shown. There are at most three special machines on~$G$: $m$, $m'$ and the last machine on~$W$, if $W$ is not a single cycle. Machines not on $W$ remain unchanged. Other machines on $W$ reallocate some allocation to better jobs than before. Thus, they develop their lexicographical situation and keep their refusal edges. Machines with the same set of refusal edges do not gain new possibly proposal edges, while machines with a better lexicographical situation do not gain new blocking edges. Thus, all new edges in $P \cup P'$ must be incident to one of the three machines mentioned above. Due to the last operation along $W$ that shifted~$r(j)$, $m$ possibly ceases to have refusal edges at all, thus, it may loose its possibly blocking edges. Regarding the last machine on~$W$, it loses allocation, and through that it may receive at most $|J|-1$ new blocking edges. These edges were all possibly blocking before, moreover, they lead back to~$W$. Similarly, $m'$, gaining a refusal edge, may become the end vertex of new possibly blocking edges, but there are at most $|J|-1$ of them. To summarize this: when $\Theta_2$ develops in one element, $\Theta_1$ may be increased in at most one element by at most~$|J|-1 < |V|$.

This argumentation shows that the number of iterations can be bounded by $O(|V||E|)$ from above, because $\Theta(x)$ cannot have more different states during the execution of the algorithm. Next, we determine how much time is needed to execute a single augmentation. Procedure \textsc{FindWalk} starts with choosing any machine that has a blocking edge of type~I. This can be done in $O(|V|)$ time. Adding the best proposal edge and the refusal pointer takes constant time, if they are stored for each vertex. Since at most one vertex is visited twice by the walk, after $O(|V|)$ steps, $W$ is chosen. Then, either of the two augmenting procedures is called. It modifies $x$ on $O(|V|)$ edges. At last, $R, P,$ and $P'$ are updated. As explained above, the change in those sets involves at most $O(|V|)$ edges.

In total, the algorithm performs $O(|V||E|)$ rounds, each of them needs $O(|V|)$ time to be computed. In total, the accelerated Phase~I algorithm runs in $O(|V|^2|E|)$ time.

\subsection*{Accelerated second phase}

The second phase can be accelerated in a very similar manner to the first phase. Instead of describing this new algorithm directly and proving its correctness using the same methods as above, we choose a shortcut. The main idea in this subsection is that the accelerated second phase of our algorithm is actually the accelerated first phase of the same algorithm on a slightly modified instance. Thus, its correctness and running time have already been proved.

At the beginning of our argumentation we make these modifications on the instance $\mathcal{I}$ given at the termination of the accelerated Phase~I algorithm. We show that the set of blocking edges of type~I on the modified instance $\mathcal{I}'$ and the set of blocking edges of type~II on $\mathcal{I}$ coincide. Then we let our accelerated Phase~I algorithm run on~$\mathcal{I}'$. At the end, we argument that its output is stable on~$\mathcal{I}$.

\subsubsection*{Modified instance.}

After the termination of the first phase, an allocation $x_0$ is given so that all blocking edges are of type~II. This input of the second phase is modified the following way. A dummy job $j_d$ and edges between each machine and $j_d$ are added to~$G$. The capacity of these edges equals the maximum quota amongst all machines, $q(j_d)$ is their sum. While $j_d$'s preference list can be chosen arbitrarily, the new edges stand at the bottom of the preference lists of the machines. The new graph is called~$G'$. Not only the graph, but also the allocation $x_0$ is slightly modified: machines with not yet filled up quota assign all their free quota to~$j_d$. In this new allocation,~$x_0'$, all machines are saturated. The new instance $\mathcal{I}'$ consists of $G', q', c', O'$ and~$x_0'$.

As mentioned above, our goal is to perform Phase~I operations on~$\mathcal{I}'$. In order to be able to do so, we swap the two color classes: jobs play a passive role, while machines become the active players. Since each active vertex has a filled up quota, all blocking edges are of type~I on~$\mathcal{I}'$.

\begin{figure}[H]
\begin{center}
\begin{tikzpicture}[scale=1, transform shape]

\tikzstyle{vertex} = [circle, draw=black]
\tikzstyle{edgelabel} = [circle, fill=white]

\node[vertex] (j_3) at (0, 3) {$j_3$};
\node[above=0.2 cm of j_3] {$1.9$};
\node[vertex] (j_4) at (4, 3) {$j_4$};
\node[above=0.2 cm of j_4] {$1$};
\node[vertex] (j_2) at (-4, 3) {$j_2$};
\node[above=0.2 cm of j_2] {$1$};
\node[vertex] (j_1) at (-8, 3) {$j_1$};
\node[above=0.2 cm of j_1] {$1$};

\node[vertex] (m_1) at (-6,0) {$m_1$};
\node[below=0.2 cm of m_1] {$2.8$};
\node[vertex] (m_2) at (-2, 0) {$m_2$};
\node[below=0.2 cm of m_2] {$1$};
\node[vertex] (m_3) at (2,0) {$m_3$};
\node[below=0.2 cm of m_3] {$1$};

\draw [] (j_3) -- node[edgelabel, near start] {2} node[edgelabel, near end] {2} (m_1);
\draw [ultra thick] (j_3) -- node[edgelabel, near start] {1} node[edgelabel, near end] {2} (m_3);
\draw [] (j_3) -- node[edgelabel, near start] {3} node[edgelabel, near end] {1} (m_2);
\draw [ultra thick] (j_4) -- node[edgelabel, near start] {1} node[edgelabel, near end] {2} (m_2);
\draw [] (j_4) -- node[edgelabel, near start] {2} node[edgelabel, near end] {1} (m_3);

\draw [ultra thick] (j_2) -- node[edgelabel, near end] {1} (m_1);
\draw [ thick] (j_1) -- node[auto=right] {$0.8$} node[edgelabel, near end] {3} (m_1);
	
\end{tikzpicture}
\end{center}
\caption{$x_0$ on $\mathcal{I}$}
\end{figure}
\begin{figure}[H]
\begin{center}
\begin{tikzpicture}[scale=0.8, transform shape]

\tikzstyle{vertex} = [circle, draw=black]
\tikzstyle{edgelabel} = [circle, fill=white]

\node[vertex] (j_3) at (0, 3) {$j_3$};
\node[above=0.2 cm of j_3] {$1.9$};
\node[vertex] (j_4) at (4, 3) {$j_4$};
\node[above=0.2 cm of j_4] {$1$};
\node[vertex] (j_2) at (-4, 3) {$j_2$};
\node[above=0.2 cm of j_2] {$1$};
\node[vertex] (j_1) at (-8, 3) {$j_1$};
\node[above=0.2 cm of j_1] {$1$};
\node[vertex] (j_d) at (8, 3) {$j_d$};
\node[above=0.2 cm of j_d] {$4.8$};

\node[vertex] (m_1) at (-6,0) {$m_1$};
\node[below=0.2 cm of m_1] {$2.8$};
\node[vertex] (m_2) at (-2, 0) {$m_2$};
\node[below=0.2 cm of m_2] {$1$};
\node[vertex] (m_3) at (2,0) {$m_3$};
\node[below=0.2 cm of m_3] {$1$};

\draw [gray, ultra thick] (j_d) -- node[edgelabel, near end] {4} (m_1);
\draw [gray] (j_d) -- node[edgelabel, near end] {3} (m_2);		
\draw [gray] (j_d) -- node[edgelabel, near end] {3} (m_3);

\draw [] (j_3) -- node[edgelabel, near start] {2} node[edgelabel, near end] {2} (m_1);
\draw [ultra thick] (j_3) -- node[edgelabel, near start] {1} node[edgelabel, near end] {2} (m_3);
\draw [] (j_3) -- node[edgelabel, near start] {3} node[edgelabel, near end] {1} (m_2);
\draw [ultra thick] (j_4) -- node[edgelabel, near start] {1} node[edgelabel, near end] {2} (m_2);
\draw [] (j_4) -- node[edgelabel, near start] {2} node[edgelabel, near end] {1} (m_3);

\draw [ultra thick] (j_2) -- node[edgelabel, near end] {1} (m_1);
\draw [ thick] (j_1) -- node[auto=right] {$0.8$} node[edgelabel, near end] {3} (m_1);
	
\end{tikzpicture}
\end{center}
\caption{$x_0'$ on $\mathcal{I'}$}
\end{figure}

Note that $\mathcal{I}'$ was constructed in such a way that - regardless of the type of blocking- each edge blocking $x$ also blocks $x'$ and vice versa. This is due to the fact that the only difference between the two instances is that machines' free quota appears as allocation on their worst edge on~$\mathcal{I}'$. The definition of a blocking edge does not distinguish between those two notions. In particular, given a specific allocation $x_0$ with no blocking edge of type~I, the set of Phase~II blocking edges on $\mathcal{I}$ and the set of Phase~I blocking edges on $\mathcal{I}'$ trivially coincide.

Let us denote the output of the accelerated Phase~I algorithm on $\mathcal{I}'$ by~$x'$, its restriction on $E(G)$ by~$x$.

\begin{claim}
	Allocation $x$ is stable on~$\mathcal{I}$.
\end{claim}

\begin{proof}
	Suppose edge $jm$ blocks~$x$. On~$\mathcal{I}'$, $jm$ is unsaturated and dominates $x'$ at both end vertices, hence $jm$ blocks $x'$ as well. Since $x'$ is the output of the accelerated Phase~I algorithm on~$\mathcal{I}'$, $jm$ is of type~II. Our goal is to show by induction that $x'(m) = q(m)$ for all machines. Thus, a contradiction is derived, because on~$\mathcal{I}'$, no Phase~II blocking edge may occur.

	At start, $x'(m) = q(m)$ for all machines. The key property of $x_0'$ is that all unsaturated edges that dominate $x_0'$ at their (saturated) machine are not better than their job's worst edge in~$x_0'$. Otherwise, they would be blocking edges of type~I to~$x_0$. Adding a blocking edge $jm$ to $x_0'$ can therefore never result in a refusal by the passive vertex~$j$. Thus, after the first round, $x'(m) = q(m)$ still holds. Alternating walks are chosen in such a manner that jobs increase $x'$ only on their best proposal edges. This guarantees that even after the first round, if $jm$ dominates the current allocation $x_1'$ at $m$, it is not better than $j$'s worst edge in~$x_1'$. From this on, induction applies.
\end{proof}

The runtime of this phase may not exceed the runtime of the accelerated Phase~I algorithm, since the size of $\mathcal{I}'$ does not exceed the size of $\mathcal{I}$ significantly.

\begin{algorithm}[H]
\renewcommand{\thealgorithm}{}
\label{alg:acc:phI}
\caption{Accelerated Phase~I}
\begin{algorithmic}
	\While{$|P| > 0$}
		\State \Call{FindWalk}{H(x)}
		\If{$W$ is a cycle}
			\State \Call{AugmentCycle}{W}
		\Else
			\State \Call{AugmentWalk}{W}
		\EndIf		
	\State update $R$, $P$, $P'$
	\EndWhile
\end{algorithmic}
\end{algorithm}

\begin{algorithmic}[]
	\Procedure{FindWalk}{H(x)}
	\label{FindWalk}
		\State $i := 1$, $W := \emptyset$, find any $m_1 \in M$ with a $P$-edge
		\While{$m_i$ has a proposal edge and $m_i$'s best proposal edge $j_i m_i \cap W = \emptyset$}
			\State $W := W \cup j_i m_i \cup r(j_i)$
			\State $j_i m_{i+1} := r(j_i), i := i+1$
		\EndWhile
	\EndProcedure
\end{algorithmic}
\vline

\begin{algorithmic}[]
	\Procedure{AugmentCycle}{$W$}
		\State $A := \min\{x(r(j)), \bar{x}(p), \bar{x}(p')| r(j) \in W \cap R, p, p' \in W \cap (P  \cup P')\}$
		\For{$p, p' \in W \cap (P  \cup P')$}
			\State $x(p) := x(p) + A, x(p') := x(p') + A$
		\EndFor
		\For{$r(j) \in W \cap R$}
			\State $x(r(j)) := x(r(j)) - A$
		\EndFor
	\EndProcedure
\end{algorithmic}
\vline

\begin{algorithmic}[]
	\Procedure{AugmentWalk}{$W$}
		\State $A := \min \{x(r(j)), \bar{x}(p), \bar{x}(p'), \bar{x}(m_1) + x(\text{edges dominated by }j_1 m_1 \text{ at }m_1) | r(j) \in W \cap R, p, p' \in W \cap (P  \cup P')\}$
		\If{$A - \bar{x}(m_1)$ > 0}
			\State $m_1$ refuses $A - \bar{x}(m_1)$ allocation from its worst edges
		\EndIf
		\For{$p, p' \in W \cap (P  \cup P')$}
			\State $x(p) := x(p) + A, x(p') := x(p') + A$
		\EndFor
		\For{$r(j) \in W \cap R$}
			\State $x(r(j)) := x(r(j)) - A$
		\EndFor		
	\EndProcedure
\end{algorithmic}

Recall the example instance again. Checking both proposal and both refusal edges on $W = (m_1 j_3, j_3 m_2, m_2j_4, j_4m_3)$, the residual capacity of $m_1$, and the allocation on $m_1$'s worse edges, it turns out that~$A = 1$. Thus, allocation $x$ shown on the figure above is obtained after the first augmentation. While $j_3m_1$ leaves~$P$, $j_3m_3$ enters it. The set of refusal edges consists of all edges with positive allocation value. $P'$ is empty. In the second round, $W$ is easy to find: it is~$(m_3j_3, j_3m_1)$. After reassigning allocation of value~1 to~$j_3m_3$, Phase~I ends. The allocation derived is not yet stable: $j_1m_1$ and $j_3m_2$ block it, but they are both of type~II. 

\begin{figure}[H]
\begin{center}
\begin{tikzpicture}[scale=1, transform shape]

\tikzstyle{vertex} = [circle, draw=black]
\tikzstyle{edgelabel} = [circle, fill=white]

\node[vertex] (j_3) at (0, 3) {$j_3$};
\node[above=0.2 cm of j_3] {$1.9$};
\node[vertex] (j_4) at (4, 3) {$j_4$};
\node[above=0.2 cm of j_4] {$1$};
\node[vertex] (j_2) at (-4, 3) {$j_2$};
\node[above=0.2 cm of j_2] {$1$};
\node[vertex] (j_1) at (-8, 3) {$j_1$};
\node[above=0.2 cm of j_1] {$1$};

\node[vertex] (m_1) at (-6,0) {$m_1$};
\node[below=0.2 cm of m_1] {$2.8$};
\node[vertex] (m_2) at (-2, 0) {$m_2$};
\node[below=0.2 cm of m_2] {$1$};
\node[vertex] (m_3) at (2,0) {$m_3$};
\node[below=0.2 cm of m_3] {$1$};

\draw [ultra thick] (j_3) -- node[edgelabel, near start] {2} node[edgelabel, near end] {2} (m_1);
\draw [] (j_3) -- node[edgelabel, near start] {1} node[edgelabel, near end] {2} (m_3);
\draw [] (j_3) -- node[edgelabel, near start] {3} node[edgelabel, near end] {1} (m_2);
\draw [ultra thick] (j_4) -- node[edgelabel, near start] {1} node[edgelabel, near end] {2} (m_2);
\draw [] (j_4) -- node[edgelabel, near start] {2} node[edgelabel, near end] {1} (m_3);

\draw [ultra thick] (j_2) -- node[edgelabel, near end] {1} (m_1);
\draw [ thick] (j_1) -- node[auto=right] {$0.8$} node[edgelabel, near end] {3} (m_1);
	
\end{tikzpicture}
\end{center}
\caption{After the first round of the accelerated Phase~I algorithm}
\end{figure}
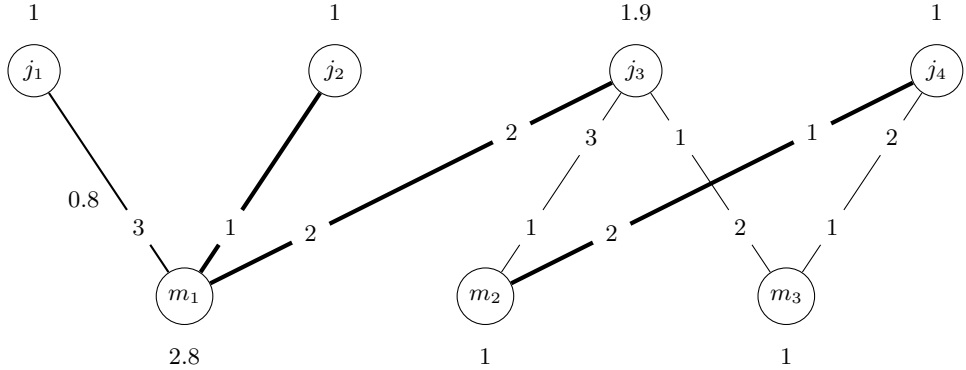

Our method resembles the well-known notion of \emph{rotations}~\cite{irvingbook}. They can be used when deriving a stable solution from another, by finding an alternating cycle of matching and non-matching edges and augmenting along them. In our algorithm, when we are searching for augmenting cycles or walks, we use an approach similar to rotations: jobs candidate their edges better than their worst positive edge, while machines choose the best out of them. However, two differences can be spotted right away. While rotations are always assigned to a stable solution different from the job-optimal, our method works on unstable input. Moreover, besides cycles we also augment along paths and walks.

\section*{Conclusion and open questions}

We solved the problem of uncoordinated processes on stable allocation instances algorithmically. Our first method is a deterministic better response algorithm that finds a stable solution through executing myopic steps. In case of rational input data, the existence of such an algorithm guarantees that random better response strategies terminate with a stable solution with probability one. Analogous results are shown for best response dynamics. We also prove that random best response strategies terminate in expected polynomial time on correlated markets, even in the presence of irrational data. An accelerated version of our first algorithm is provided as well. For any real-valued instance, it terminates after $O(|V|^2|E|)$ steps with a stable allocation. We also show a counterexample for a possible acceleration for the case of best response dynamics.

Future research may involve more complex stability problems from the paths-to-stability point of view. For example, any of the problems listed after the figure in Section~\ref{subsec:uncoord} can be combined with our general setting of allocations.

\subsection*{Acknowledgements}

We would like to thank P\'eter Bir\'o and Tam\'as Fleiner for their valuable comments. This work was supported by the Deutsche Telekom Stiftung, the Berlin Mathematical School (BMS), and the Deutsche Forschungsgemeinschaft within the research training group `Methods for Discrete Structures' (GRK 1408).

\bibliographystyle{abbrv}
\bibliographystyle{unsrt}
\bibliography{mybib}

\end{document}